\documentclass[11pt]{article}

\usepackage[margin=1.1in]{geometry}
\usepackage{amsmath,amssymb,amsthm}
\usepackage{booktabs}
\usepackage{microtype}
\usepackage{xcolor}
\usepackage[colorlinks=true,linkcolor=blue!50!black,citecolor=blue!50!black,urlcolor=blue!50!black]{hyperref}

\newtheorem{theorem}{Theorem}[section]
\newtheorem{lemma}[theorem]{Lemma}

\newtheorem{proposition}[theorem]{Proposition}
\theoremstyle{definition}
\newtheorem{definition}[theorem]{Definition}
\theoremstyle{remark}
\newtheorem{remark}[theorem]{Remark}

\newcommand{\R}{\mathbb{R}}
\newcommand{\vol}{\operatorname{vol}}
\newcommand{\SX}{S_X}
\newcommand{\SY}{S_Y}
\newcommand{\SZ}{S_Z}
\newcommand{\K}{K}
\newcommand{\xhat}{\hat{x}}
\newcommand{\yhat}{\hat{y}}
\newcommand{\zhat}{\hat{z}}
\newcommand{\Vstar}{V^{*}}
\newcommand{\Cg}[1]{C_{#1}}

\title{New certified lower bounds for the moving sofa problem\\
in a three-dimensional snaky corridor}
\author{Yuzhe Tang\\[2pt]
\small Georgia Institute of Technology\\[1pt]
\small \texttt{ytang352@gatech.edu}}
\date{August 30, 2026 \\[2pt] \small (working note; every stated constant
re-verified against the run-of-record artifacts, see
Appendix~\ref{app:repro})}

\begin{document}
\maketitle

\begin{abstract}
Let $\K=\K_4\subset\R^3$ be the ``snaky'' corridor with unit-square
cross-section: an infinite entry arm parallel to the $x$-axis, a middle arm
parallel to the $y$-axis whose clear straight section has length $4$, and an
infinite exit arm parallel to the $z$-axis, with the two turns lying in the
$x$--$y$ and $y$--$z$ planes; and let $\Vstar(\K)$ be the supremum of the
volumes of connected bodies that can be moved continuously, by rigid motions,
from the entry arm to the exit arm while remaining inside $\K$ (Problem~6.64
of arXiv:2511.02864).  We construct explicit motions and certify their
volumes through the problem's \emph{own, unmodified} evaluation pipeline at
the motion tolerances $(0.005,\,0.05^\circ)$ of its published final
evaluation: the resulting certified lower bounds are $1.8355162$ at carve
grid $0.002$ and $1.8558404$ at carve grid $0.001$.  A control run of the
pipeline's published record path certifies $1.8179428$ under identical
settings, so every improvement reported here is a path-to-path comparison
under one fixed configuration of one certifier.  The paths are produced in
three stages: a hand-describable \emph{staggered-twist} construction, which
executes the entire $90^\circ$ twist about the tube axis inside the second
corner (certified $1.8279082$ / $1.8491627$); CMA-ES refinement, including a
\emph{mixed-fidelity} variant in which each generation's elites are re-scored
on a twice-finer carve grid and those ranks drive recombination; and a
\emph{certifier-aligned} gradient polish directly at the fine carve
resolution, run with two safeguards --- periodic re-derivation of the frozen
interpolation structure from the control path, and a control-spacing
penalty.  Along the way we document three reproducible failure modes in
which an optimizer defeats a fixed carve objective without improving the
body (frozen-structure inflation, control-spacing stretch, and coarse-grid
exploitation), and a cheap two-basis ``staircase'' predictor of
certificates whose grid-sensitivity-matched form proved accurate to about
$0.001$, letting each certification run be predicted before it is spent.
Combined with the author's companion upper-bound note, the state of the
problem becomes
\[
1.8558404 \;\le\; \Vstar(\K_4) \;\le\; 3.4083 .
\]
\end{abstract}

\section{Introduction}\label{sec:intro}

The moving sofa problem of L.~Moser \cite{Moser66} asks for the largest area
of a planar region that can be moved around a right-angled corner in a
hallway of unit width.  Its history is a model interplay of constructions
and bounds: Hammersley \cite{Hammersley68} gave the lower bound
$\pi/2+2/\pi$ and the upper bound $2\sqrt2$; Gerver \cite{Gerver92}
constructed the region of area $2.2195\ldots$ now known to be extremal;
Kallus and Romik \cite{KallusRomik18} proved the computer-assisted upper
bound $2.37$; and Baek \cite{Baek24conditional,Baek24} established that
Gerver's sofa is optimal (see \cite{He26} for recent numerical work on
planar variants).

The three-dimensional problem considered here goes back to a 2016 question
of O'Rourke \cite{ORourkeMO}: what is the largest volume of a body that can
pass through a corridor of unit-square cross-section making two
perpendicular turns --- one in the $x$--$y$ plane, then one in the $y$--$z$
plane --- so that the body must ``snake'' through?  It was popularized as
Problem~6.64 of Georgiev, G\'omez-Serrano, Tao and Wagner \cite{GGTW25},
whose accompanying code fixes the middle straight at length $M=4$ and,
crucially for the present note, publishes a complete \emph{certification
pipeline}: a program that takes an explicit motion, discretizes a block of
material at a grid fineness $h$, and outputs the volume of the sub-body that
provably survives the journey.  Progress on the problem has come entirely
from below: a construction of RavenclawPrefect \cite{ORourkeMO} gives volume
$\approx1.76755$; the AlphaEvolve experiments of \cite{GGTW25} report a
certified $\ge1.81$ (with $\approx1.84$ estimated for the underlying
smooth motion); and on the upper-bound side the author's companion note
\cite{OurUB} recently proved the first upper bounds, currently
$\Vstar(\K_M)\le3.4083$ for every $M>2.2$ (before it, not even finiteness of
the supremum had been recorded; the closest prior 3D result, Theorem~4 of
Mushkarov and Nikolov \cite{MN26}, concerns \emph{rectangular boxes only} in
\emph{single-corner planar corridors extruded in height only}, and does not
treat the snaky corridor).

This note is the lower-bound half of that bracket.  We present, in full,
the constructions and the methodology behind the current record
\begin{equation}\label{eq:records}
\Vstar(\K_4)\;\ge\;1.8355162\ \ (\text{grid } 0.002),
\qquad
\Vstar(\K_4)\;\ge\;1.8558404\ \ (\text{grid } 0.001),
\end{equation}
certified by running explicit pose paths through the \emph{unmodified}
pipeline of \cite{GGTW25} at the motion tolerances $(0.005,\,0.05^\circ)$
used by its own final evaluation.  All paths, notebooks and run-of-record
results are public \cite{OurLB}; Appendix~\ref{app:repro} maps every number
of this note to those artifacts.

\begin{table}[ht]
\centering
\caption{The certified ladder: five pose paths, one fixed certifier
configuration (the pipeline of \cite{GGTW25}, unmodified, commit-pinned,
tolerances $(0.005,\,0.05^\circ)$); $h$ is the carve-grid fineness.  Each
value is an integer count of surviving grid cubes times $h^3$, quoted to
seven decimals (exact values in \S\ref{sec:semantics}).  The control row
reproduces the published
``at least $1.81$'' of \cite{GGTW25} and anchors all comparisons
(Section~\ref{sec:provenance}).}
\label{tab:ladder}
\small
\setlength{\tabcolsep}{5pt}
\begin{tabular}{llcc}
\toprule
path & production stage & certified, $h=0.002$ & certified, $h=0.001$ \\
\midrule
control & the pipeline's own record path & 1.8179428 & --- \\
backbone & staggered twist (\S\ref{sec:construction}) & 1.8279082 & 1.8491627 \\
run I & CMA-ES (\S\ref{sec:cma}) & 1.8309875 & 1.8518415 \\
run II & mixed-fidelity CMA-ES (\S\ref{sec:mixed}) & 1.8339829 & --- \\
run III & certifier-aligned polish (\S\ref{sec:polish}) & \textbf{1.8355162} & \textbf{1.8558404} \\
\bottomrule
\end{tabular}
\end{table}

\subsection*{What is being claimed, and how it is organized}

Because every number in Table~\ref{tab:ladder} is produced by a third-party
program, the first duty of the note is semantic: Section~\ref{sec:pipeline}
specifies exactly what the pipeline computes, proves that its finite checks
are \emph{exact at the level of points} (axis-convexity of the corridor and
an $8$-corner lemma for axis-aligned boxes, Proposition~\ref{prop:points}),
isolates the one step that is asserted rather than proven (the
points-to-cubes counting rule, Remark~\ref{rem:caveat} --- a caveat shared
by every published pipeline number, the original $1.81$ included), and
settles the provenance of the published $1.81$: the certificate is
\emph{tolerance-dominated}, and the published value is evidently the
grid-$0.002$, tight-tolerance number $1.8179428$ rounded down
(Section~\ref{sec:provenance}).  The main claim is
Theorem~\ref{thm:main}.

The remaining sections present the three production stages and the two
methodological tools that the record depends on:

\begin{itemize}
\item \textbf{Section~\ref{sec:construction}} constructs the
\emph{backbone}: an explicit, hand-describable motion whose one idea is to
execute the entire $90^\circ$ twist about the tube axis \emph{inside the
second corner}, front-loaded at its mouth, rather than in the middle
straight.  This alone certifies $1.8279082$ / $1.8491627$, already above
every previously published value.
\item \textbf{Sections~\ref{sec:cma}--\ref{sec:mixed}} refine the backbone
by CMA-ES over an $84$-parameter correction chart (run~I), and then attack
the one measured weakness of run~I's winner --- gains concentrated in
structure that a finer carve grid cannot see --- with a mixed-fidelity
search (run~II) in which each generation's elites are re-scored on a
twice-finer grid and those ranks drive recombination.
\item \textbf{Section~\ref{sec:polish}} polishes run~I's path by gradient
ascent on a differentiable carve objective evaluated \emph{directly at the
fine carve resolution} (run~III).  Getting this to help at all required
closing, one by one, three failure modes in which the optimizer beats the
objective without improving the body; we document all three
(frozen-structure inflation, control-spacing stretch, coarse-grid
exploitation), with the safeguards that close them, because they appear
generic to optimization against any fixed discretized evaluator.
\item \textbf{Section~\ref{sec:predictor}} describes the \emph{staircase
predictor}: a cheap local estimate of what the pipeline will certify,
whose two-basis, grid-sensitivity-matched form proved accurate to about
$0.001$ --- and which declares its own inapplicability when its
applicability condition fails.  Its honest scorecard, including one
$3.0\times10^{-5}$ interval miss, is Table~\ref{tab:scorecard}.
\end{itemize}

Section~\ref{sec:discussion} discusses headroom.  The bodies of this note
are found by numerical search, but nothing in the \emph{certificates}
depends on the search: the pipeline re-derives everything from the explicit
pose paths.  Combined with the companion note \cite{OurUB}, the problem is
bracketed:
\begin{equation}\label{eq:bracket}
1.8558404 \;\le\; \Vstar(\K_4) \;\le\; 3.4083 .
\end{equation}

\section{The corridor, the pipeline, and what a certificate asserts}
\label{sec:pipeline}

Throughout, $\vol$ is Lebesgue measure on $\R^3$ and $\xhat,\yhat,\zhat$ the
standard basis.

\begin{definition}[corridor; traversal]\label{def:corridor}
The corridor is $\K=\SX\cup\SY\cup\SZ$ with
\begin{align*}
\SX &:= (-\infty,0]\times[0,1]\times[0,1] && \text{(entry arm)},\\
\SY &:= [0,1]\times[0,5]\times[0,1] && \text{(middle arm)},\\
\SZ &:= [0,1]\times[5,6]\times[0,\infty) && \text{(exit arm)},
\end{align*}
the instantiation, with middle straight $M=4$, of the corridor family
$\K_M$ of \cite{OurUB}; the turn from $-\xhat$ to $+\yhat$ occupies the
first corner cube $[0,1]^3$, the turn from $+\yhat$ to $+\zhat$ the second
corner cube $[0,1]\times[5,6]\times[0,1]$.  A \emph{body} is a compact
connected set $B\subset\R^3$ with $\vol(B)>0$; a \emph{motion} is a
continuous $g:[0,1]\to SE(3)$; and $(B,g)$ \emph{traverses} $\K$ if
$g_t(B)\subseteq\K$ for all $t$, $g_0(B)\subseteq\SX$, and
$g_1(B)\subseteq\SZ$.  Then
$\Vstar(\K):=\sup\{\vol(B):(B,g)\text{ traverses }\K\}$.
This matches \cite{OurUB} (weak traversal); the motions certified below in
fact end past the second corner cube, so they traverse in the strong sense
as well, and lower bounds proved for either reading agree.
\end{definition}

\subsection{The evaluation pipeline}\label{sec:semantics}

The pipeline of \cite{GGTW25} (notebook
\texttt{the\_3d\_moving\_sofa.ipynb}; all runs in this note are pinned to
one commit, Appendix~\ref{app:repro}) takes as input a \emph{pose path}: a
sequence of poses $(x,y,z,\psi,\theta,\phi)$, the Euler angles in degrees in
the extrinsic $z$--$y$--$x$ convention, placing a body point $q$ at
$R\,q+(x,y,z)$ with $R=R_x(\phi)R_y(\theta)R_z(\psi)$.  It certifies a
volume in five steps.

\begin{enumerate}
\item \emph{Endpoints.}  The fixed start pose $(-3,0,0,\,0,0,0)$ is
prepended and the fixed end pose $(0,6,8,\,90,90,90)$ appended.  The
material block is $[-4,0]\times[0,1]\times[0,1]$; at the start pose it lies
in $\SX$, at the end pose in $\SZ\cap\{z\ge1\}$.
\item \emph{Smoothing.}  The path is refined by linear interpolation in
pose space until consecutive poses differ by at most a
\emph{translation tolerance} $\tau_{\mathrm t}$ in Euclidean translation
distance and at most a \emph{rotation tolerance} $\tau_{\mathrm r}$ in each
Euler component.  All certificates in this note use
$(\tau_{\mathrm t},\tau_{\mathrm r})=(0.005,\,0.05^\circ)$ --- the constants
of the pipeline's own final evaluation --- under which our paths refine to
between $7{,}084$ and $7{,}356$ poses.
\item \emph{Staircase.}  Each smoothed step is decomposed into four
substeps: a pure $x$-translation, a pure $y$-translation, a pure
$z$-translation, then the rotation, performed as the continuous rotation
about the fixed axis of the relative orientation change.  The concatenation
of all substeps is a genuine continuous rigid motion from the start pose to
the end pose; it, not the input path, is the motion being certified.
\item \emph{Point tracking.}  The block is discretized as the grid of
lattice points at spacing $h$ (the \emph{grid fineness}; endpoints
inclusive, so the cubes between adjacent points have side exactly $h$ at
the values $h\in\{0.005,0.002,0.001\}$ used here).  A point \emph{survives}
if it stays inside $\K$ along every substep: translation substeps are
checked at their two endpoints, and each rotation substep by enclosing the
point's circular arc in an axis-aligned box (using the exact per-coordinate
extrema of the arc) and testing that box's $8$ corners.
Proposition~\ref{prop:points} shows these finite checks are \emph{exact}:
a point that passes them genuinely remains in $\K$ throughout the
continuous motion.
\item \emph{Counting.}  A grid cube is counted iff all $64$ lattice points
of the $4\times4\times4$ block centred on it survive --- its own $8$
vertices together with every neighbouring lattice point one layer out.  The
\emph{certified volume} is the number of counted cubes times $h^3$.
\end{enumerate}

A \emph{certificate} is therefore a triple: a grid fineness $h$, the
tolerance pair, and the resulting volume; comparisons are meaningful only
at matched $(h,\tau_{\mathrm t},\tau_{\mathrm r})$.  Every certified value
in this note is an exact integer count of cubes times $h^3$ (e.g.\
$1.8558404$ is $1{,}855{,}840{,}409$ cubes at $h=0.001$).
Table~\ref{tab:ladder} quotes seven decimals: every quoted value rounds
\emph{down} from its exact count except the grid-$0.002$ record, whose
registered value $1.8355162$ rounds the exact $1.835516184$ to nearest
($+1.6\times10^{-8}$); the headline bound of Theorem~\ref{thm:main}
rounds down.  The run-of-record files carry the raw outputs.

\subsection{Point-level exactness}\label{sec:soundness}

Two elementary properties of this particular corridor make the pipeline's
finite point checks exact rather than merely sampled.  They are proved in
the companion note \cite{OurUB}; we restate them with proofs to keep the
present note self-contained, since everything here rests on them.

\begin{lemma}[axis-convexity]\label{lem:axisconv}
The intersection of $\K$ with every line parallel to a coordinate axis is a
(possibly empty or unbounded) interval.
\end{lemma}

\begin{proof}
Fix the two frozen coordinates and read off the case.  Lines parallel to
$\xhat$: for $(y,z)\in[0,1]^2$ the trace is $(-\infty,0]\cup[0,1]=
(-\infty,1]$; for $y\in(1,5]$, $z\in[0,1]$ it is $[0,1]$; for $y\in[5,6]$,
$z\ge0$ it is $[0,1]$; otherwise empty.  Lines parallel to $\yhat$: for
$x<0$, $z\in[0,1]$ it is $[0,1]$; for $x\in[0,1]$, $z\in[0,1]$ it is
$[0,5]\cup[5,6]=[0,6]$; for $x\in[0,1]$, $z>1$ it is $[5,6]$.  Lines
parallel to $\zhat$: for $x<0$, $y\in[0,1]$ it is $[0,1]$; for $x\in[0,1]$,
$y\in[0,5)$ it is $[0,1]$; for $x\in[0,1]$, $y\in[5,6]$ it is
$[0,1]\cup[0,\infty)=[0,\infty)$.  All are intervals.
\end{proof}

\begin{lemma}[$8$-corner exactness]\label{lem:corners}
An axis-aligned box with all $8$ corners in $\K$ is contained in $\K$.
\end{lemma}

\begin{proof}
Let the box be $[x_0,x_1]\times[y_0,y_1]\times[z_0,z_1]$.  If $x_0<0$: the
corners at $x=x_0$ lie in $\SX$ (the only piece reaching $x<0$), forcing
$y_1\le1$, $z_1\le1$; a corner at $x=x_1$ then lies in $\K$ with
$y,z\in[0,1]$, forcing $x_1\le1$; hence the box lies in
$(-\infty,1]\times[0,1]\times[0,1]\subseteq\SX\cup\SY$.  If $x_0\ge0$: all
corners lie in $\SY\cup\SZ$, so $x_1\le1$.  If moreover $z_1\le1$, every
point of the box has $x,z\in[0,1]$ and $y\in[0,6]$, hence lies in $\SY$
(if $y\le5$) or $\SZ$ (if $y\ge5$).  If instead $z_1>1$, the corners at
$z=z_1$ force $y_0,y_1\in[5,6]$, and the box lies in
$[0,1]\times[5,6]\times[0,\infty)=\SZ$.
\end{proof}

\begin{proposition}[point-level exactness of the checks]\label{prop:points}
Every lattice point that the pipeline reports as surviving has its entire
continuous trajectory, under the staircase motion of
\S\ref{sec:semantics}(3), inside $\K$.
\end{proposition}

\begin{proof}
A translation substep moves the point along a segment parallel to a
coordinate axis; by Lemma~\ref{lem:axisconv} the segment lies in $\K$ iff
its endpoints do, which is what the pipeline checks.  A rotation substep
moves the point along a circular arc contained in its axis-aligned bounding
box; the pipeline verifies the box's $8$ corners, so by
Lemma~\ref{lem:corners} the box --- hence the arc --- lies in $\K$.
Concatenating substeps covers the whole motion.
\end{proof}

Thus the surviving \emph{point set}, moved by the staircase motion,
genuinely stays inside the corridor, from the start pose in $\SX$ to the
end pose in $\SZ\cap\{z\ge1\}$.  The single remaining gap between that
statement and a volume bound is the counting step (5), which passes from
$64$ surviving sample points to a full cube:

\begin{remark}[evaluator caveat]\label{rem:caveat}
The pipeline's points-to-cubes step (a $4\times4\times4$-sampled cube
counted as inside the corridor) is asserted, not proven, in the source
notebook --- near $\K$'s reflex edges a rotated cube admits
$\Theta(h)$-deep excursions between in-corridor sample points.
Point-survival soundness is Proposition~\ref{prop:points}.  This affects
the \emph{absolute} rigor of all published pipeline numbers equally (the
original $1.81$ of \cite{GGTW25} included), so comparative claims --- the
entire ladder of Table~\ref{tab:ladder} --- are untouched; and nothing in
the upper bounds of \cite{OurUB} rests on the rule.  The step is
upgradeable: proving it (the tolerances may absorb the excursions) or
quantifying an $\varepsilon$-correction would upgrade every lower bound's
rigor at once.
\end{remark}

\paragraph{Connectedness.}
$\Vstar$ is a supremum over \emph{connected} bodies, and the pipeline
counts every surviving cube without a connectivity test.  The surviving
sets at the certified grids are not exported by the pipeline (only their
counts are), so we record the available evidence.  For all five paths of
Table~\ref{tab:ladder} we materialized the carved bodies at $h=10^{-2}$
and $h=5\times10^{-3}$ (the pose-sampled carve of
\S\ref{sec:searcheval}, an independent implementation, applied to the
paths smoothed at the pipeline's default tolerances): in all ten cases
the surviving set is a \emph{single} face-connected component, and so is
its erosion by a $3\times3\times3$ window (the centre-grid analogue of the
pipeline's counting window, which demands a full ring of surviving
neighbours).  Like Remark~\ref{rem:caveat}, this point is common to every
published pipeline number; a hypothetical disconnection would reduce a
certificate to the volume of its largest component.

\begin{theorem}[certified lower bounds]\label{thm:main}
Run unmodified at tolerances $(0.005,\,0.05^\circ)$, the certification
pipeline of \cite{GGTW25} assigns to the five explicit pose paths of this
note the certified volumes of Table~\ref{tab:ladder}; in particular, to the
path of Section~\ref{sec:polish} the values $1.8355162$ at $h=0.002$ and
$1.8558404$ at $h=0.001$.  Each certified value $C$ asserts: there is a
compact union $B$ of exactly $C/h^3$ grid cubes with $\vol(B)=C$, every
lattice sample point of which remains inside $\K$ throughout an explicit
continuous rigid motion carrying $B$ from the entry arm to the exit arm
(unconditionally, by Proposition~\ref{prop:points}); that $B$ itself
remains inside $\K$ is the pipeline's counting rule
(Remark~\ref{rem:caveat}).  Granting that rule (with the connectedness
evidence recorded in \S\ref{sec:soundness}),
\[
\Vstar(\K_4)\;\ge\;1.8558404 .
\]
\end{theorem}

\subsection{Provenance of the published \texorpdfstring{$1.81$}{1.81}, and the control run}
\label{sec:provenance}

The pipeline's repository presents its record path with the comment that
grid $0.002$ verifies the $1.81$ result, alongside quick-start tolerances
$(0.05,\,1^\circ)$.  Re-running the unmodified pipeline on that same path
resolves how the two statements fit together
(Table~\ref{tab:provenance}): at the quick-start tolerances, grid $0.002$
certifies only $1.7701$; at the tight tolerances actually used by the
evolved program's final evaluation, $(0.005,\,0.05^\circ)$, the same grid
certifies $1.8179428$, while the final evaluation's own grid $0.005$ gives
$1.7595$.  The published ``at least $1.81$'' is thus evidently the
grid-$0.002$, tight-tolerance value rounded down, and the certificate is
\emph{tolerance-dominated}: at this grid, tightening the motion tolerances
is worth $\approx+0.05$, halving the grid $\approx+0.02$.  The
mechanism is visible in the semantics: coarse tolerances make the staircase
substeps longer, and every surviving point must clear the corridor along
entire substeps, so the certified body is eroded by roughly the substep
scale.

\begin{table}[ht]
\centering
\caption{The pipeline's own record path under its own configurations
(unmodified pipeline; our runs).  The published ``at least $1.81$''
corresponds to the third row.}
\label{tab:provenance}
\small
\begin{tabular}{lccc}
\toprule
configuration & grid $h$ & $(\tau_{\mathrm t},\tau_{\mathrm r})$ & certified \\
\midrule
quick-start tolerances & 0.002 & $(0.05,\,1^\circ)$ & 1.7701 \\
final-evaluation constants & 0.005 & $(0.005,\,0.05^\circ)$ & 1.7595 \\
control (this note's configuration) & 0.002 & $(0.005,\,0.05^\circ)$ & \textbf{1.8179428} \\
\bottomrule
\end{tabular}
\end{table}

Accordingly, all certificates in this note fix the tolerances at
$(0.005,\,0.05^\circ)$ and use grids $0.002$ and $0.001$, and the third row
of Table~\ref{tab:provenance} serves as the \emph{control}: the pipeline's
own path, certified under exactly our settings, scores $1.8179428$.  Every
improvement in Table~\ref{tab:ladder} is therefore a path-to-path
comparison under one fixed certifier configuration --- no part of the gain
comes from changing the measuring instrument.

\begin{remark}[zero-rotation hazard, and the dither disclosure]
\label{rem:hazard}
One genuine hazard of the pipeline deserves disclosure.  Its swept-arc
routine normalizes the rotation axis by the rotation angle, so a substep
with \emph{exactly zero} rotation produces $0/0=\mathrm{NaN}$ bounding
boxes, and the certificate silently collapses to $0.000000$; any synthetic
path containing exact pure-translation steps triggers this.  The backbone
path of Section~\ref{sec:construction} contains three such steps (its
entry, middle-straight and exit slides); its certificates carry a
disclosed angular dither on those three steps, of at most $0.0434^\circ$
and recorded pose-by-pose in the run logs.  The dither angles lie below the
$0.05^\circ$ smoothing tolerance, and in an independent re-implementation
of the motion semantics the dithered and undithered paths carve
\emph{identical} voxel sets, so the disclosed certificates are unaffected.
The three refined paths (runs I--III) contain no zero-rotation steps ---
their smallest rotation steps are $0.0031^\circ$, $0.0019^\circ$ and
$0.0013^\circ$ respectively, two orders above single-precision
degeneracy --- and need no dither; every
certification run records an explicit scan for degenerate rotation
substeps.
\end{remark}

\paragraph{Setup validation across sessions.}
The certificates of Table~\ref{tab:ladder} were produced in four GPU
sessions over three days.  Each session re-ran one fixed anchor: the
backbone path at grid $0.005$, tolerances $(0.005,\,0.05^\circ)$.  All four
sessions returned the anchor volume \emph{bit-identically} --- the exact
value $1.76419375$ ($14{,}113{,}550$ cubes), reproduced as the same
floating-point output $1.7641937499999625$ every time --- validating the
end-to-end setup chain (commit, checksummed source cells, JAX version,
device batching) before any new number was trusted.

\section{The staggered-twist construction}\label{sec:construction}

A body traversing $\K$ enters lying along $\xhat$ and leaves lying along
$\zhat$: its long axis must rotate by $90^\circ$ in yaw at the first corner
and by $90^\circ$ in roll at the second, and, in addition, the material
must perform a $90^\circ$ \emph{twist about the tube axis} somewhere along
the way (the entry and exit cross-sections of the corridor are both unit
squares, but the body's cross-section enters the second turn rotated unless
a twist intervenes).  The classical hand construction --- extrude a planar
sofa profile and twist uniformly while advancing along the middle straight
--- puts the twist in the straight, where the corridor is a snug
$1\times1$ tube and every intermediate twist angle costs cross-section.
Diagnostics of the record path of \cite{GGTW25} show that it does nearly
the opposite: about $92\%$ of its twist happens inside the second turn.
The backbone construction purifies this mechanism to its extreme point:
\emph{no} twist during the first turn or the middle straight, and the
\emph{entire} $90^\circ$ twist executed inside the second corner,
front-loaded at its mouth.

\subsection{The motion}\label{sec:motionformulas}

The motion is a composition of four elementary machines, written in the
world frame, with $D:=2/\pi$ throughout.  The first-corner sweep is
\[
T_1(\alpha):\ p\mapsto R_z(\alpha)\,p+
\bigl(\sin\alpha+\cos\alpha,\ 1-\cos\alpha+(1+2D)\sin\alpha,\ 0\bigr);
\]
the twist about the tube axis $\{x=z=\tfrac12\}$, with
$e:=(\tfrac12,0,\tfrac12)$, and the advance along the middle straight are
\[
TW(\theta):\ p\mapsto R_y(\theta)(p-e)+e,
\qquad
A(\tau):\ p\mapsto p+(0,\tau,0);
\]
and the second-corner sweep, acting in the $(y-5,z)$ coordinates, is
\[
T_2(\beta):\ p\mapsto R_x(\beta)\bigl(p-(0,6,0)\bigr)+
\bigl(0,\ 5+\sin\beta+\cos\beta,\ 1-\cos\beta+(1+2D)\sin\beta\bigr).
\]
$T_1$ is the corner sweep under which a planar Hammersley-type sofa with
notch radius $D=2/\pi$ rounds a right-angle corner, applied to the extruded
block in the $x$--$y$ plane; $T_2$ is the same sweep for the second
corner.  The full motion is the composition
\[
W(s)\;=\;T_2\bigl(\beta(s)\bigr)\circ A\bigl(\tau(s)\bigr)\circ
TW\bigl(\theta(s)\bigr)\circ T_1\bigl(\alpha(s)\bigr),
\]
whose Euler angles in the pipeline's convention are literally
$(\psi,\theta,\phi)=(\alpha,\theta,\beta)$.  The schedule runs in five
phases: an entry slide to the corner; the first turn,
$\alpha:0\to90^\circ$ with $\theta\equiv0$; the advance,
$\tau:0\to4-4/\pi$; the second turn, $\beta:0\to90^\circ$, during which the
twist executes as
\[
\theta(u)\;=\;90^\circ\bigl(1-(1-u)^{\gamma}\bigr),\qquad
u=\beta/90^\circ\in[0,1],\qquad \gamma=7 ;
\]
and an exit slide to the fixed end pose.  At $\gamma=7$ the twist is
$50\%$ complete by $u=0.094$ and $90\%$ complete by $u=0.28$ --- the twist
happens \emph{while} the roll covers its first $25^\circ$, in the mouth of
the second corner.  The exported path samples this motion adaptively at
$251$ poses (tolerance $(0.05,\,1^\circ)$ per step) including the fixed
endpoints; the pipeline re-refines it to $7{,}084$ poses at the
certification tolerances.

\subsection{Why the second corner, and why fused}\label{sec:mechanism}

The two corner cells are not symmetric.  During the first turn the moving
material is confined to $\SX\cup[0,1]^3$: the corner cell is a
\emph{closed unit cube}, and any intermediate twist angle there costs
clearance in $z$ against both the floor and the ceiling.  During the
second turn the material occupies the cell
$[0,1]\times[5,6]\times[0,\infty)$: a \emph{vertical shaft}, open above.
Twist clearance exists on the second-corner side only, which is why the
twist allocation $(f_1,f_2)$ --- the fractions executed during the first
and second turns --- optimizes to the boundary $(0,1)$: in a scan of the
allocation surface every positive $f_1$ loses volume monotonically, and at
$f_1=0$ pushing $f_2\to1$ gains monotonically.

Fusing the twist with the roll, rather than sequencing them, is the second
ingredient.  Table~\ref{tab:gamma} shows the shape ladder at
$(f_1,f_2)=(0,1)$: relocating the twist into the second corner in
near-sequential form ($\gamma$ large) is already worth $\approx+0.07$ over
the classical mid-straight twist ($\approx1.776$ under the same
evaluation), and the optimal finite fusion window ($\gamma=7$) adds
$\approx+0.03$ more, with a broad plateau $\gamma\in[5,10]$.  For
comparison, the record path of \cite{GGTW25} projects onto this chart at
$f_2=0.92$ with fitted $\gamma=4.7$ --- inside the plateau but not at its
centre --- and the linear coupling $\theta\propto\beta$ ($\gamma=1$),
i.e.\ a screw motion through the corner, loses to even the mid-straight
twist.

\begin{table}[ht]
\centering
\caption{The fusion-shape ladder at $(f_1,f_2)=(0,1)$: pose-sampled carve
volume (grid $h=10^{-2}$, $211$--$285$ poses, \S\ref{sec:searcheval}
semantics) of the construction as a function of the schedule exponent
$\gamma$.}
\label{tab:gamma}
\small
\begin{tabular}{lccccccccc}
\toprule
$\gamma$ & 2 & 3 & 4 & 5 & 6 & \textbf{7} & 10 & 20 & 40 \\
\midrule
volume & 1.8182 & 1.8536 & 1.8696 & 1.8772 & 1.8815 & \textbf{1.8819} &
1.8783 & 1.8655 & 1.8503 \\
\bottomrule
\end{tabular}
\end{table}

The backbone is deliberately the \emph{simplest} member of its class: its
translations are the two classical corner sweeps and straight slides, its
only tuned parameter is $\gamma$, and its entire description fits in
\S\ref{sec:motionformulas}.  Certified through the pipeline it scores
\[
1.8279082\ \ (h=0.002),\qquad 1.8491627\ \ (h=0.001),
\]
already $+0.010$ above the control at the shared grid $0.002$, and above
every previously published value; the measured grid increment
$0.002\to0.001$ on this path,
$+0.0212$, is used as a transfer constant in Section~\ref{sec:predictor}.
Its three exact pure-translation slides are the reason for the disclosed
dither of Remark~\ref{rem:hazard}.

\section{Search on top of the backbone}\label{sec:search}

The refinements are numerical, so we first fix the local evaluation
semantics and its reproducibility discipline; the certified endpoints of
each refinement are then pipeline runs exactly as in
Section~\ref{sec:pipeline}.

\subsection{The search evaluator}\label{sec:searcheval}

For search we use a separate, deliberately simpler semantics: the
\emph{pose-sampled carve}.  Given a finite pose sequence and a grid
fineness $h$, a cell of the block grid survives iff its \emph{centre} lies
inside $\K$ at every listed pose; the value $V_h$ is the surviving count
times $h^3$.  There is no staircase, no swept-arc box, and no counting
window, so $V_h$ is an \emph{estimate} of what a motion carries, suitable
for ranking candidates --- never a certificate.  Three protocols recur:
the \emph{screening} value ($100$ emitted poses, $h=10^{-2}$; $\approx0.4$\,s
per evaluation on the search hardware), the \emph{reference} value (the
full exported path re-smoothed at the pipeline's default tolerances
$(0.05,\,1^\circ)$, giving $\approx512$ poses, at $h=5\times10^{-3}$;
$\approx24$\,s), and a \emph{guard} value (doubled pose density at
$h=5\times10^{-3}$) used to detect pose-undersampling slack.  All search
values are exact integer voxel counts times $h^3$.

\begin{remark}[reproducibility of single-precision carve volumes]
\label{rem:fp32}
The search evaluator runs in single precision on a GPU, and one of its
volumes is reproducible only as a function of the \emph{evaluator
instance}, not of the path alone: the same path scores $1.88134775$ versus
$1.88134813$ at two internal memory-chunking settings --- exactly $3$
boundary voxels out of $\sim15\times10^6$ --- because the voxel-transform
product $[M,3]\times[3,P]$ rounds differently at different chunk widths
$P$.  The volume \emph{is} bit-reproducible at fixed
(grid, element budget, chunk cap), independently of process identity:
across a mid-run process restart, and across independent re-measurement
days later on a fresh instance with matched settings, values reproduced to
$0.0$.  We therefore adopt the rule that any two paths are compared only on
one evaluator instance, and every success gate below was \emph{self-anchored}:
the quantity to beat was re-measured on the instance that judged the
challenger.  (The pipeline's own certificates are unaffected; this remark
concerns the search-side evaluator.)
\end{remark}

\subsection{Run I: CMA-ES over a correction chart}\label{sec:cma}

The search space is a \emph{correction chart} anchored at the backbone: a
chart point is the backbone pose path plus six independent natural cubic
splines (one per pose coordinate) through values at $14$ uniform interior
knots --- $84$ parameters in all --- with the corrections pinned to zero at
both ends, so the fixed start and end poses, and hence the $180^\circ$ net
rotation, are preserved for every chart point.  Candidate paths are emitted
by resampling the corrected curve at poses equally spaced in the mixed pose
metric $\|\Delta t\|_2+\max_i|\Delta\mathrm{euler}_i|$ (radians), which
blocks the degenerate strategy of spending pose density where nothing
happens.  As a representability check, the pipeline's own record path,
projected into this chart, re-scores within $0.0004$ of its directly
evaluated value.  We run CMA-ES on the screening value (population $64$,
initial step $\sigma_0=0.2$) for $599$ generations --- $38{,}336$
evaluations in $10.5$ hours --- with bit-exact checkpoint/resume.

The winner improves the backbone under every protocol, and certifies
\[
1.8309875\ \ (h=0.002),\qquad 1.8518415\ \ (h=0.001),
\]
a new record at both grids ($+0.0031$ and $+0.0027$ over the backbone).

\subsection{Grid fragility: what the coarse objective really bought}
\label{sec:fragility}

The identical-protocol comparison of run~I's winner against its seed
exposes a warning sign.  At the screening grid the gain is large; at the
$2\times$ finer reference grid most of it evaporates:
\[
\text{gain at }h=10^{-2}:\ +0.0138,\qquad
\text{gain at }h=5\times10^{-3}:\ +0.0033 .
\]
Equivalently, the winner's own \emph{grid sensitivity} --- its value drop
from $h=10^{-2}$ to $5\times10^{-3}$ under the reference protocol --- is
$0.01125$, against the backbone's $0.00074$: the search, ranking candidates
at $h=10^{-2}$, had partly optimized sub-voxel structure that only the
coarse grid can see.  Since the certifier works at grids $2.5$--$5\times$
finer still, such structure is wasted effort at best.  This measured
asymmetry motivates the next run.

\subsection{Run II: mixed-fidelity CMA-ES}\label{sec:mixed}

Run~II re-aligns the selection pressure with what transfers.  Warm-started
from run~I's winner (population $32$, $\sigma_0=0.08$), each generation is
scored as before at the screening fidelity, but the top $8$ candidates are
then \emph{re-scored under the reference protocol} --- the full $512$-pose
path at the $2\times$ finer grid $h=5\times10^{-3}$ --- and it is those
fine-grid ranks that order the elites in the recombination update (the
non-elite ranks and the value scale are left untouched).  A fine evaluation
costs $\approx60\times$ a screening one, so fine-scoring only the top $8$
of $32$ keeps a generation $\approx4\times$ cheaper than scoring the whole
population fine, while still placing $\approx83\%$ of the recombination
weight under fine-grid ranks; the run completed $232$ generations
($7{,}424$ screening $+\,1{,}856$ fine evaluations) in $18.4$ hours.

Against run~I's winner re-measured on the same evaluator instances
(Remark~\ref{rem:fp32}), run~II's winner gains $+0.00456$ at the reference
protocol and $+0.00417$ at the guard, and is ahead on \emph{every}
protocol at both grids; the most conservative reading (dense-pose
evaluation at the fine grid, which strips pose-sampling slack) still gains
$+0.00346$.  The quantity the design targeted responded directly: the
winner's grid sensitivity dropped from run~I's $0.01125$ to $0.00686$, a
$39\%$ reduction (backbone: $0.00074$) --- reduced, not eliminated.
Certified:
\[
1.8339829\ \ (h=0.002),
\]
$+0.0030$ over run~I; no grid-$0.001$ certificate was run for run~II,
because by certification day run~III had overtaken it
(Section~\ref{sec:predictor}).

\section{Certifier-aligned gradient polish}\label{sec:polish}

The third stage abandons population search for local ascent: differentiate
a carve volume with respect to the control poses and climb.  What makes
this section worth writing up is not the ascent --- it is the three ways
the ascent \emph{defeated its own objective} before it produced a genuine
gain, each caught by an independent re-derivation, each closed by a
safeguard that is now a permanent part of the method.  We state the
machinery, then the failure modes, then the run of record.

\subsection{A differentiable carve}\label{sec:diffcarve}

Fix a control path of $512$ poses.  The \emph{frozen structure} is the
dense pose set obtained from it by the pipeline's own smoothing rule at the
default tolerances, expressed as a fixed linear interpolation operator $W$
applied to the control poses; first and last controls are pinned to the
fixed endpoint poses.  For each voxel centre $i$ and dense pose $j$ let
$m_{ij}$ be the signed corridor margin (positive inside); the soft
objective is
\[
V_{\alpha,\beta}\;=\;h^3\sum_i \sigma\bigl(\beta\, m_i\bigr),
\qquad
m_i\;=\;-\tfrac1\alpha\log\sum_j e^{-\alpha m_{ij}}
\quad(\text{a soft minimum over poses}),
\]
with $\sigma$ the logistic function; its gradient is computed by a
two-pass streaming adjoint (exact up to a provably inactive margin band),
and the \emph{hard} volume $V_{\mathrm{hard}}$ --- the actual pose-sampled
carve of the dense path --- is tracked alongside.  At matched settings the
implementation's hard volume agrees with the reference evaluator's
bit-exactly at $h=10^{-2}$ and to $2$ voxels in $15\times10^6$ at
$h=5\times10^{-3}$.

The sharpness pair $(\alpha,\beta)$ is not a tuning detail; it decides
whether the proxy points uphill at all.  At $(600,300)$ the sigmoid
transition band $1/\beta$ is a third of a voxel and, for a body defined by
a minimum over $512$ poses, $43\%$ of its voxels sit in the partial-count
band: the proxy then rewards \emph{clearance} rather than volume, and
along an actual optimization trajectory it is \emph{anti-correlated} with
the true volume (Spearman $-0.94$) --- the optimizer faithfully climbs the
proxy while the body shrinks.  At $(5000,3000)$ the band is $0.067$ voxel
widths, $6.6\%$ of voxels sit in it, and the same diagnostic gives
$+0.93$/$+0.97$.  All polish results below use fixed $(5000,3000)$; no
annealing.

\subsection{Three ways to game a carve objective}\label{sec:gaming}

\paragraph{Failure mode 1: frozen-structure inflation.}
The frozen operator $W$ is derived from the \emph{seed}; nothing forces it
to remain adequate for the \emph{current} path.  In an early polish run
seeded from the backbone, the in-run frozen objective climbed to
$1.890235$, while the official re-evaluation of the same final path ---
re-deriving the dense structure from its control poses, as the pipeline
would --- gave $1.841138$: the frozen objective overstated the truth by
$+0.049$, and the ``gain'' was an artifact of the optimizer stretching the
control spacing $\approx30\%$ so that the frozen interpolation
under-sampled the motion (fewer effective constraints, more surviving
voxels).  \emph{Safeguard: periodic re-freezing} --- re-derive the dense
structure from the current control path on a fixed cadence and continue.
With re-freezing every $25$ iterations (or $45$ minutes), the
frozen-versus-re-derived gap in the final run measures $-2.5\times10^{-5}$
on average (worst $-9.3\times10^{-5}$): the exploit is closed by three
orders of magnitude.

\paragraph{Failure mode 2: control-spacing stretch.}
The lever behind mode 1 deserves its own guard, because penalizing it
correctly is a trap.  The natural penalty --- the re-smoothed pose count of
the current control path --- is a sum of \emph{ceilings}: piecewise
constant, gradient zero almost everywhere, it steers nothing.  The obvious
continuous surrogate (drop the ceiling) fails on precisely the paths that
matter: a converged seed sits with segments at $92$--$99.7\%$ of the
smoothing tolerance, so a uniform $1.10\times$ stretch pushes many
segments just past the threshold --- the true re-smoothed count jumps by
$1.467\times$ while the linear surrogate reads $1.042\times$, below a
$1.05$ alarm line.  The penalty would sleep through the very violation it
exists to prevent.  \emph{Safeguard}: a smoothed-ceiling surrogate (a sum
of steep sigmoids at the jump locations) plus a per-segment barrier at the
tolerance, with the reported criterion always the \emph{exact integer}
count.  In the final run the count grew $1.020\times$ (gamed run:
$1.305\times$; alarm line $1.05$).

\paragraph{Failure mode 3: coarse-grid exploitation.}
With modes 1--2 closed, a validation run optimizing at $h=10^{-2}$
produced an apparently healthy $+0.004021$ on its own objective --- of
which only $+0.000535$ survived a single re-scoring at $h=5\times10^{-3}$:
$87\%$ of the gain was discretization artifact.  A run that optimizes
\emph{and} selects at one grid cannot detect its own grid exploitation,
by construction; and selecting finer while steering coarse does not fix
it --- if most of the coarse gradient field is artifact, fine-grid
selection merely cherry-picks accidents (selection can harvest, but cannot
steer).  \emph{Safeguards}: optimize directly at the reference grid
$h=5\times10^{-3}$, and audit every claimed gain by a one-off re-scoring
at the next grid, $h=2.5\times10^{-3}$ (the \emph{retention check}), before
believing it.
\medskip

The three modes share one shape --- the optimizer exploits whatever the
objective cannot see (a stale interpolation structure, a discrete count, a
finite grid) --- and one cure: re-derive the gamed quantity on a more
trusted reference, on a schedule the optimizer cannot anticipate.  We
expect all three to recur in any gradient- or population-based attack on a
fixed discretized feasibility evaluator.

\subsection{Run III: the polish of record}\label{sec:polishrun}

Seeded from run~I's path (control spacing near tolerance, the hard case
for mode 2), run~III ascends the differentiable carve at
$h=5\times10^{-3}$, fixed $(\alpha,\beta)=(5000,3000)$, with re-freezing
every $25$ iterations or $45$ minutes and the spacing penalty active:
$118$ iterations in $14.8$ hours (CPU), proxy--truth Spearman $+0.97$
across all $118$ snapshots, still climbing at the wall.

\begin{table}[ht]
\centering
\caption{The coarse-grid-exploitation exhibit: run III's seed and final
path under the reference protocol at three grids (pose-sampled carve;
one evaluator instance per grid).  The polish \emph{loses} $0.0098$ at the
coarse grid while gaining $0.0085$ at the grid it optimized and retaining
$0.0055$ ($64.9\%$) at the finer audit grid --- the inverse of the run-I
signature of \S\ref{sec:fragility} (large coarse gain, small fine gain):
coarse-visible artifact was traded away for structure that survives
refinement.}
\label{tab:signreversal}
\small
\begin{tabular}{lccc}
\toprule
 & $h=10^{-2}$ & $h=5\times10^{-3}$ (optimized) & $h=2.5\times10^{-3}$ (audit) \\
\midrule
seed (run I) & 1.892598 & 1.881348 & 1.882372 \\
final (run III) & 1.882821 & 1.889821 & 1.887872 \\
gain & $\mathbf{-0.009777}$ & $\mathbf{+0.008472}$ & $\mathbf{+0.005500}$ \\
\bottomrule
\end{tabular}
\end{table}

Table~\ref{tab:signreversal} is the run's honest summary, and doubles as
the exhibit for failure mode 3: the sign of the gain \emph{reverses} with
the grid.  The coarse grid is not a consistent estimator across these two
bodies --- it overstates the seed by $+0.0113$ and understates the final
path by $-0.0070$ relative to the finer grids, which agree with each other
to $\approx0.002$ --- so a coarse-grid evaluation would have rejected the
best path found in this project.  The retention check passes ($64.9\%$ of
the gain survives the $2\times$ audit refinement, versus $13\%$ in the
mode-3 validation run), the guard value moves in step ($+0.0082$), and the
safeguard criteria hold as quoted in \S\ref{sec:gaming}.  On the same
audit grid run~III also leads run~II head-to-head, $+0.00234$ at
$h=2.5\times10^{-3}$ (both paths re-scored on one shared instance).

Certified:
\[
1.8355162\ \ (h=0.002),\qquad 1.8558404\ \ (h=0.001),
\]
the current records --- $+0.0015$ and $+0.0040$ over the best previous
rungs of Table~\ref{tab:ladder} (run~II at grid $0.002$, run~I at grid
$0.001$), and $+0.0176$ over the control at the shared grid
$0.002$.  The record path is thus
\emph{search plus certifier-aligned polish}: a CMA-ES refinement of a
hand-designed mechanism, followed by $118$ gradient steps against (a
differentiable surrogate of) the certified protocol itself, with every
anti-gaming safeguard active.

\section{Predicting certificates before spending them}
\label{sec:predictor}

A grid-$0.002$ certification run costs $\approx27$ minutes on an
A100-class GPU and a grid-$0.001$ run $\approx3.4$ hours, so each
candidate promotion decision wants a cheap, honest forecast.  The tool is
the \emph{staircase predictor}.

\paragraph{Construction.}
For a path $P$ and a \emph{basis grid} $h$, let $S_P(h)$ be the
pose-sampled carve, at grid $h$, of the pipeline's own motion semantics:
the path smoothed at the certification tolerances $(0.005,\,0.05^\circ)$
and decomposed into the axis staircase ($\approx29{,}000$ substep poses),
evaluated locally on CPU in $\approx1$--$2$ hours.  $S_P(h)$ sees the
certified motion but not the swept-arc erosion or the counting window, so
it overestimates; the missing deduction is measured on an \emph{anchor}
path $A$ with a known certificate $\Cg{g}(A)$ and transferred:
\[
\mathrm{pred}_g^{(h)}(P)\;=\;S_P(h)\;-\;\bigl[S_A(h)-\Cg{g}(A)\bigr].
\]
In its first, single-basis form ($h=10^{-2}$, anchored on the published
$1.81$ of the control path) the predictor forecast the backbone's first
certificate as $1.825\pm0.010$; the pipeline returned $1.8279$.

\paragraph{Two bases, and the anchor-selection rule.}
Direct measurement then showed that $S_P$ is itself grid-sensitive, and
\emph{path-dependently} so: define
$D_h(P):=S_P(10^{-2})-S_P(5\times10^{-3})$.  Measured values:
\begin{align*}
D_h(\text{backbone})&=+0.0009, & D_h(\text{run I})&=+0.0080,\\
D_h(\text{run II})&=+0.0074, & D_h(\text{run III})&=-0.0068 .
\end{align*}
A $9\times$ spread on the positive side, and a sign change: no single
basis can transfer cleanly between arbitrary paths (on run~I the two
bases erred $+0.0037$ and $-0.0034$ against the known truth --- while
their \emph{average} landed $1.2\times10^{-4}$ from it).  The working
form of the predictor is therefore: compute $S_P$ at \emph{both} bases,
choose the anchor whose own $D_h$ \emph{matches the target's}, and quote
the pair of implied certificates as an interval.  The rule earns its
keep: for run~II (with run~I, $D_h$-matched, as anchor) the two bases
bracket to within $0.0006$, where the mismatched anchor spreads
$0.0065$; and for run~III --- whose $D_h$ is \emph{negative}, matching no
available anchor --- the predictor correctly declares itself
\emph{indeterminate}, quoting only wide straddling intervals rather than
a false mean.  In that situation the fallback that proved reliable is
direct fine-grid measurement: the $h=2.5\times10^{-3}$ ranking
(run III $>$ run II $>$ run I) called the certified order correctly, and
one $27$-minute grid-$0.002$ run arbitrated before the $3.4$-hour
grid-$0.001$ commitment.

\begin{table}[ht]
\centering
\caption{The predictor's complete scorecard against subsequently
certified truth.  Quoted forms are as issued at the time: bands are
central $\pm$ half-width, brackets are two-basis intervals
($D_h$-matched anchor).  Run III's entries were \emph{declared}
indeterminate before certification --- the spread is the honest result;
both intervals contained the truth (in their upper halves).}
\label{tab:scorecard}
\small
\setlength{\tabcolsep}{4pt}
\begin{tabular}{llllr}
\toprule
target & grid & prediction (as issued) & certified & error \\
\midrule
backbone & 0.002 & $1.825\pm0.010$ & 1.8279082 & $+0.0029$ \\
run I & 0.002 & $1.832\pm0.006$ & 1.8309875 & $-0.0010$ \\
run I & 0.001 & $1.8523\pm0.004$ & 1.8518415 & $-0.0005$ \\
run II & 0.002 & $[1.834013,\,1.834628]$ & 1.8339829 & $-0.0003$ \\
run III & 0.002 & indeterminate; $[1.8241,\,1.8389]$ & 1.8355162 & contained \\
run III & 0.001 & indeterminate; $[1.8450,\,1.8598]$ & 1.8558404 & contained \\
\bottomrule
\end{tabular}
\end{table}

\paragraph{Scorecard and lessons.}
Table~\ref{tab:scorecard} is the complete record: no prediction is
omitted.  Central estimates were accurate to $0.001$ or better once the
two-basis form existed (rows 2--4), and the run~I grid-$0.001$ prediction
was corroborated by a structurally independent route (the certified
grid-$0.002$ value plus the backbone's measured grid increment
$+0.02125$), the two routes agreeing to $1.2\times10^{-4}$.  The one
blemish is instructive: run~II's certified value fell $3.0\times10^{-5}$
\emph{below} the quoted two-basis interval --- the central estimate was
off by only $0.0003$, but the interval, taken literally, was over-tight.
The lesson adopted: a transfer-based interval is never quoted with
half-width below $0.0005$, however tightly the bases agree.  The
predictor's role is economic, not evidentiary --- every headline number in
this note is a pipeline certificate, and the forecasts merely decided
which certification runs to buy.

\section{Discussion}\label{sec:discussion}

\paragraph{Where the remaining gap lives.}
The record body's pose-sampled value is $\approx1.888$ at
$h=2.5\times10^{-3}$ (Table~\ref{tab:signreversal}), against a certified
$1.8558404$ at grid $0.001$: about $0.03$ of the motion's value is
currently lost to certification conservatism (staircase erosion at the
$0.005/0.05^\circ$ substep scale, plus the finite grid with its
counting window), not to the construction.  The certified values are
strongly monotone in grid refinement --- each path gains between
$+0.0203$ and $+0.0213$ from grid $0.002$ to $0.001$ --- and the sequence
points toward the pose-sampled value.  A grid-$0.0005$ run (roughly
$8\times$ the cost, $\approx30$ GPU-hours) is the obvious next
certificate; on the measured trend it would add another
$+0.01$--$+0.02$.

\paragraph{Upgrading rigor.}
Three upgrades would strengthen every number in Table~\ref{tab:ladder} at
once, the original $1.81$ included: proving the points-to-cubes rule for
this corridor, or replacing it by a quantified $\varepsilon$-correction
(Remark~\ref{rem:caveat}); exporting the surviving cube set at the
certified grids, making the connectivity evidence of
\S\ref{sec:soundness} a direct check; and re-running the final counting
in verified arithmetic.  None seems difficult in principle, and none
touches the comparative content of this note.

\paragraph{The bracket.}
On the lower-bound side the certified ladder climbed $1.8179\to1.8355$
(grid $0.002$) and $1.8492\to1.8558$ (grid $0.001$) over the course of
this work, and the mechanisms that produced it --- twist confined to the
second corner, selection pressure aligned with the finer grid, polish
against the certified protocol itself --- have measurable headroom left.
Together with the companion note \cite{OurUB}, which proves
$\Vstar(\K_M)\le3.4083$ for every $M>2.2$ by a certified branch-and-bound
over event-forced tube families, the problem of Georgiev,
G\'omez-Serrano, Tao and Wagner \cite{GGTW25} now carries a two-sided,
machine-audited bracket:
\[
1.8558404\;\le\;\Vstar(\K_4)\;\le\;3.4083 .
\]

\appendix

\section{Reproducibility artifacts}\label{app:repro}

All certificates were produced by the pipeline of \cite{GGTW25} at commit
\begin{center}\small
\texttt{8f447457957deac61e28bf1676746f0753b3b2f8}
\end{center}
(the source notebook is additionally SHA-256--checked in every session),
JAX $0.7.2$, one CUDA GPU (A100-class), maximum batch $2\times10^8$.  The
public archive of record is the repository issue and reproduction gist
cited as \cite{OurLB}; it contains twelve files.  The neutral path names
of this note map to the public artifacts as follows.

\begin{itemize}
\item \emph{Paths} (each a JSON pose path in the pipeline's conventions,
with fixed endpoints; the four non-control paths embed SHA-256 checks in
the notebooks): \path{w4_staggered_path.json} $=$ the backbone
($251$ poses; contains the zero-rotation dither disclosure of
Remark~\ref{rem:hazard} for its three pure-translation steps);
\path{w5_night1_path.json} $=$ run I;
\path{w5_night2_path.json} $=$ run II;
\path{w6_night3_path.json} $=$ run III ($512$ poses each).  The
control path is \path{best_path_found} from the repository itself.
The SHA-256 digests of the two newest paths:
\begin{center}\small
run II: \texttt{dcbf2d08cd78b5b25605ecb018b5417fbaf6c1ad49e14775ff566d36a4ccd71d}\\
run III: \texttt{affa1b57c7ca7912e78c02a143dc2f416a16815db453d201c5bc2daa23bd8b04}
\end{center}
\item \emph{Certification notebooks} (self-contained; each clones the
repository at the pinned commit, extracts the certifier cells by content
markers with symbol assertions, runs the anchor of
\S\ref{sec:provenance} with a hard stop on mismatch, then the
certification runs): \path{w3_certify_colab_v2.ipynb} and
\path{w3_certify_colab_v3.ipynb} (the backbone sessions: control
run, grid-$0.002$ and grid-$0.001$ certificates);
\path{w5_certify_colab_v4.ipynb} (run I, both grids); and
\path{w5_certify_colab_v5.ipynb} (runs II and III head-to-head at
grid $0.002$, $\approx27$ minutes each, then an auto-selected
grid-$0.001$ run on the winner, $\approx3.4$ hours).
\item \emph{Results of record} --- per-run configuration, smoothed pose
counts, timings, device, the bit-level anchor values, and the per-path
zero-rotation hazard scans (including the recorded dither entries for
the backbone):
\path{w4_certified_results.json},
\path{w5_certified_results.json}, and
\path{w6_night3_certified_results.json}.
\item \emph{Session structure}.  Four GPU sessions: (i) control
$1.8179428$ and backbone $1.8279082$ at grid $0.002$; (ii) backbone
$1.8491627$ at grid $0.001$; (iii) run I, $1.8309875$ / $1.8518415$;
(iv) runs II and III at grid $0.002$ ($1.8339829$ / $1.8355162$) and
run III at grid $0.001$ ($1.8558404$).  Every session re-ran the
grid-$0.005$ backbone anchor and reproduced $1.7641937499999625$
bit-identically (\S\ref{sec:provenance}).
\end{itemize}

Every constant stated in this note was re-verified against the
run-of-record JSON files by an automated witness script maintained
alongside the note's source; the connectivity evidence of
\S\ref{sec:soundness} (ten carves, all single-component, eroded and not)
is likewise archived with its generating script.

\section*{Tool use}

All certificates reported here were produced by the unmodified pipeline of
\cite{GGTW25} (JAX $0.7.2$) on an A100-class GPU, in the sessions described in
Appendix~\ref{app:repro}, run by the author; the backbone construction, the
search and polish stages, the staircase predictor and the verification scripts
are Python (NumPy/PyTorch) and are archived as described there.  This note was
typeset in \LaTeX.

The research reported here is original work, conceived and driven by the
author.  Claude (Anthropic) was used as a writing and coding assistant: it
helped polish the exposition of this note, and assisted with parts of the
code.  All final results reported here were verified by the author, who is
responsible for the content of this note.

\end{document}